\documentclass[a4paper]{amsart}

\usepackage{amsfonts}
\usepackage{amsmath}
\usepackage{amsthm}
\usepackage{appendix}
\usepackage[english]{babel}
\usepackage{braket}
\usepackage{dsfont}
\usepackage[utf8]{inputenc} 
\usepackage{mathtools}
\usepackage{subcaption} 
\usepackage{tikz}
\usepackage{pgfplots}
\usepackage{xcolor}
\usepackage{xspace} 
\usepackage{booktabs}
\usepackage{stmaryrd}

\usepackage[
  pdfborder={0 0 0},
  unicode=true,
  pdfusetitle
]{hyperref}

\newcommand{\Sc}{\mathcal{S}}
\newcommand{\Pc}{\mathcal{P}_{C}} 

\newcommand{\R}{\mathbb{R}}

\newcommand{\md}{\operatorname{d}}

\newcommand{\from}{\colon}

\newcommand{\supp}[1]{\text{supp}\left\{#1\right\}}

\renewcommand{\imath}{\mathrm{i}}
\renewcommand{\Re}{\mathrm{Re}}
\renewcommand{\Im}{\mathrm{Im}}

\newcommand{\Real}{\mathbb{R}}
\newcommand{\Cc}{\mathbb{C}}

\newcommand{\ord}[2]{\mathrm{ord}\left(#1,#2\right)}

\newcommand{\res}{\qopname\relax m{\mathrm{res}}}

\newcommand{\diff}[2]{\frac{\md #2}{\md #1}}

\newcommand{\defMath}{\coloneqq}

\def\qft{QFT\xspace}

\newcommand{\spcTm}{\ensuremath{M}}
\newcommand{\sTDim}{\ensuremath{d}}

\theoremstyle{plain}
\newtheorem{theorem}{Theorem}[section]
\newtheorem{lemma}[theorem]{Lemma}
\newtheorem{prop}[theorem]{Proposition}

\newtheorem{definition}[theorem]{Definition}

\newtheoremstyle{dotless}{}{}{\itshape}{}{\bfseries}{}{ }{}
\theoremstyle{dotless}
\newtheorem*{thm*}{Theorem}

\title{Reflection Positivity in higher derivative scalar theories}
\author[F. Arici]{Francesca Arici}
\author[D. Becker]{Daniel Becker}
\author[C. Ripken]{Chris Ripken}
\author[F. Saueressig]{Frank Saueressig}
\author[W. D. van Suijlekom]{Walter D. van Suijlekom}
\address{Institute for Mathematics, Astrophysics and Particle Physics (IMAPP), Radboud University Nijmegen, Heyendaalseweg 135, 6525 AJ Nijmegen, The Netherlands}

\begin{document}

\begin{abstract}
Reflection positivity constitutes an integral prerequisite in the Osterwalder-Schrader reconstruction theorem which relates quantum field theories defined on Euclidean space to their Lorentzian signature counterparts. In this work we rigorously prove the violation of reflection positivity in a large class of free scalar fields with a rational propagator. This covers in particular higher-derivative theories where the propagator admits a partial fraction decomposition as well as degenerate cases including e.g. $p^4$-type propagators. 
\end{abstract}

\maketitle

\tableofcontents

\section{Introduction, main result, and physics implications}\label{sec:intro}
The requirement of unitarity provides stringent constraints on the structure of a fundamentally admissible theory. 
A prototypical example for this
is provided by higher-derivative gravity which despite being perturbatively renormalizable \cite{Stelle:1976gc}
is not considered as a viable quantum theory for gravity due to unitarity violation \cite{Stelle:1977ry}. 
At the level of constructive quantum field theory (QFT) on a \emph{Lorentzian} manifold a precise notion
of a QFT has been given in form of the Wightman axioms \cite{Wightman:1956zz}. The Osterwalder--Schrader
reconstruction theorem \cite{OstSch73} then allows to construct a QFT satisfying
the Wightman axioms from a probabilistic theory defined on a \emph{Euclidean} manifold $\spcTm$, provided
that the latter satisfies the Osterwalder--Schrader axioms. 
Since, in the latter case, the functional calculus is under much better control, due to positive definiteness of the inner product on $\spcTm$ and ellipticity of the Laplace operator, probability theories play a crucial role in the treatment of \qft{}s. 

In general the Osterwalder--Schrader axioms are defined for a generic partition function $Z[J]$.
While the axioms of Euclidean invariance, analyticity and regularity do not refer to the concept of time, reflection
positivity equips $\spcTm$ with a preferred direction which is essential in
the transition to Lorentzian spacetime. In this paper we investigate the central Osterwalder--Schrader axiom, reflection positivity, for a class of free scalar theories with a positive, analytic 
\emph{covariance operator} $C\equiv C(\Delta)$ where $\Delta \equiv - \partial^2$ is the $d$-dimensional Laplacian.
 In this case $\ln Z[J] = - \left\langle J, C J \right\rangle/2$
can be expressed as a suitable scalar product including the covariance operator $C$, which drastically simplifies the analysis.
Given the fact that reflection positivity has been proven
for notoriously few theories including the free Klein--Gordon operator $(\Delta + m^2)^{-1}$ or
the Dirac operator (see for instance \cite{OstSch73,JaRi08}), this constitutes a valid starting point.

The main result of this article is Theorem \ref{thm:main} which provides a criterion on $C$ to be reflection positive and thus gives rise to a unitary and vacuum stable \qft.

The central ingredient thereby is the operation of time reflection $\theta\from \Real^\sTDim \to \Real^\sTDim$, given by $(t,\vec{x}) \mapsto (-t,\vec{x})$ on a flat Euclidean spacetime $\spcTm\equiv \Real^{\sTDim}$.
Based on this, we define the operator $\Theta$ on the space of square-integrable functions $L^2(\Real^\sTDim)$ by pullback of $\theta$, i.e. for $f \in L^2(\Real^\sTDim)$, we define
\[
			\Theta f
	\defMath
			f \circ \theta
	\text{.} \nonumber
\]

We follow the conventions of \cite{OstSch73}: we define $\Sc(\mathbb{R}_+)$ to be the space of functions $f \in \Sc(\R)$ with $\supp f \subseteq \R_{+} \coloneqq [0, \infty)$
By $\Sc (\R^n_+)$ we denoted the completed topological tensor product
$\Sc(\R^n_{+}) \coloneqq \Sc(\R^{n-1}) \hat{\otimes} \Sc(\mathbb{R}_+)$, i.e., a function $f \in \Sc(\R^{n}_+) $ if $f \in \Sc(\R^{n})$ and $\supp f \subseteq \lbrace x \in
\R^n, \ x^0 \geq 0 \rbrace$.

In the case of a free theory, reflection positivity can be rephrased into the following definition which we employ in the sequel of this paper:
\begin{definition}[Reflection positivity]
	Let $C$ be a covariance operator that commutes with $\Theta$. Then $C$ is said to satisfy reflection positivity if for all $J \in \mathcal{S}(\mathbb{R}^n_{+})$, we have the inequality
	\[
				I[J]
		\defMath
				\left\langle
					J,
					C \Theta J
				\right\rangle \geq 0
		\text{.}
	\]
\end{definition}

Based on this definition we derive necessary and sufficient conditions for a large class of free propagators to satisfy reflection positivity.

\begin{thm*}{\bf \ref{thm:main}}
	Let $C$ be a real rational function which has no poles on $\Real_+$.
A necessary and sufficient condition for the operator $C(\Delta)$ to satisfy reflection positivity is that the poles of $C$ all lie on $\R_{-}$, are simple and with non-negative residue.
\end{thm*}
The key point in the proof is establishing a relation between reflection positivity of $C(\Delta)$ and its pole structure.
We find that the inner product $I[J]\equiv 	\left\langle J,	C \Theta J\right\rangle$  is a sum over the residues at the poles deformed by some $J$ and pole dependent prefactor.
Hence, the number of poles, their degree and their type --- in particular whether they are real or complex --- determine if $I[J]$ is non-negative for all $J$.

For simple real poles with non-negative residue, we prove in Proposition \ref{prop:if-part} intactness of reflection positivity, following an argument in \cite{JaRi08}.
In case of complex poles (cf. Proposition \ref{prop:complexRPVio}) as well as non-simple real poles (cf. Proposition \ref{prop:RealHigOrd}), we construct an explicit $J$ for which $I[J]$ becomes negative, hence showing violation of reflection positivity.

Our results clarify important questions related to the vacuum stability and unitarity of QFTs which are expected to be essential in any fundamental description of nature. Based on the classical theorem by Ostrogradski \cite{Ostrogradsky:1850}, it is commonly expected that (inverse) propagators $C^{-1} = p_n(\Delta)$ where $p_n$ denotes a polynomial of degree $n \ge 2$ exhibit instabilities or unitarity violation. Theorem \ref{thm:main} provides the quantum version of this classical analysis. It establishes an explicit relation between reflection positivity and the pole structure of $C(\Delta)$. Reflection positivity is retained if and only if each pole in $C$ corresponds to a standard free field propagator with a positive squared mass.
 This incorporates the analysis \cite{JaRi08,GliJa79,JaRi07I,JaRi07II}, showing 
that the Klein-Gordon operator $C(p^2) = (p^2 + m^2)^{-1}$ obeys reflection positivity. At the same time our theorem shows that theories where $C(\Delta)$ is a genuine rational function containing higher powers of the Laplacian violate reflection positivity and thus do not satisfy the prerequisites of the Osterwalder--Schrader reconstruction theorem. This comprises in particular higher derivative theories where the propagator admits a partial fraction decomposition which are unstable within the classical Ostrogradski theorem and propagators possessing poles of order higher than one. 

Notably, the case where $C(\Delta) = (\Delta + m^2)^{-1} f(\Delta)$ with $f(\Delta)$ being an entire, analytic function of the Laplacian satisfying suitable fall-off conditions
is covered only partially by our theorem. Such structures underly the program of non-local theories of gravity reviewed in \cite{Modesto:2017sdr} and may also appear
naturally in the context of noncommutative geometry \cite{Chamseddine:2006ep,Iochum:2011yq,Sui12} or Asymptotic Safety \cite{Niedermaier:2006wt,Reuter:2012id,Roberto:book,Becker:2017tcx}.

The present work may be continued along several lines. Clearly, it would be interesting to relax our assumptions on the analytic properties of $C$ in order to cover the entire class of non-local propagators. Along different lines, one may consider non-scalar and interacting fields. The generalization for non-scalar fields may go along the same lines as in \cite{JaRi08}. Interacting fields may prove to be more problematic, since reflection positivity does not reduce to a simple positivity condition of an inner product. Finally, one may try to generalize Theorem \ref{thm:main} to non-flat spacetimes. As a first step in this direction one might consider reflection positivity on a non-flat background equipped with a foliation structure. For the case of the Klein--Gordon operator, the corresponding generalization of the Osterwalder--Schrader axioms has been considered in \cite{JaRi08,JaRi07I,JaRi07II}.

\subsection*{Acknowledgements}
We would like to thank R. Wulkenhaar for encouraging comments at the beginning of this project and R. Alkofer for many fruitful discussions. The work of D.~B., C.~R., F.~S., and W.~v.~S.\ is supported by the Netherlands Organisation for Scientific
Research (NWO) within the Foundation for Fundamental Research on Matter (FOM) grant 13VP12. F.~A and W.~v.~S. are partially supported by NWO under the VIDI-grant 016.133.326.


\section{Setting and sufficient conditions for reflection positivity}
Our starting point is a real rational function $C$ that will describe the structure of our propagator. Such a function admits a unique representation in terms of its \emph{partial fraction decomposition} (cf. \cite[Corollary 5.5.4]{Zor})
\begin{equation}
\label{eqn:polesum}
C(x) = \sum_{j=1}^N \left( \sum_{n=1}^{k_j} a_{jn}(x-z_j)^{-n}\right) + p(x),
\end{equation}
where $p$ is a real polynomial. We say that $z_j \in \Cc$ is a \emph{pole} of order $k_j$ for $C$ and we write
$\ord{C}{z_j}= k_j $.
We denote the set of poles of $C$ by $\Pc$. In the following we will assume that $C$ has no poles on the positive real line, i.e. $\Pc \cap \R_+ = \emptyset$.\footnote{In the physics setting such poles correspond to particles with a \emph{negative} squared mass. Since the existence of such a particle would allow to obtain an infinite amount of energy by creating such particles from the vacuum, this case is excluded on physical grounds.}
The coefficients appearing in the expansion can be computed as residues, by looking at $C$ as a complex function:
\begin{equation}
\label{eqn:res}
a_{jn}\coloneqq\res_{z\to z_j}	\left( \left(z-z_j\right)^{n-1}	C(z) \right).
\end{equation}
The reality condition $\overline{C(z)}=C(\overline{z})$ implies that if $z_j$ is a complex pole of order $k_j$, then $\overline{z}_j$ is also a pole of order $k_j$.

Using Fourier transforms and the pseudodifferential calculus on $\R^d$ (cf. \cite{Gr08,Str03}), together with the fact that $C$ has no poles on $\R_+$, we can write the associated covariance operator $C(\Delta)$ as a sum of powers of resolvents plus a local operator $p(\Delta)$

\begin{equation}
C(\Delta) = \sum_{j=1}^N \left( \sum_{n=1}^{k_j} a_{jn}(\Delta-z_j)^{-n}\right) + p(\Delta),
\end{equation}
acting on the Schwartz space $\Sc(\R^d_+)$.

For later reference, let us remark that the operator $C(\Delta)$ can be decomposed into the sum of simpler operators, coming from separate poles:
\begin{align*}
C_\lambda (\Delta) & =   \sum_{n=1}^{k} \left( a_{n}(\Delta-\lambda)^{-n} + \overline{a}_{n}(\Delta-\overline{\lambda})^{-n} \right)  & \lambda \in \Cc \setminus \Real, \ \ord{C}{\lambda} = k \\
C_\mu (\Delta) & =   \sum_{n=1}^{k}  a_{n}(\Delta-\mu)^{-n} & \mu  \in \R_{-}, \ \ord{C}{ \mu} = k \\
\end{align*}
Let $J \in \Sc(\R^d_{+})$. For any $w \in \Cc\setminus \R_+$, we denote by $I_{w}[J]$ the Klein--Gordon propagator with (possibly complex) mass $w$:
\begin{equation}
\label{eqn:Izj}
			I_{w}[J]
	=
			\left\langle J,	\Theta	\left(\Delta -w\right)^{-1}	J\right\rangle .
		\end{equation}

			\begin{prop}
			Let $C$ be a real rational function with partial fraction decomposition \eqref{eqn:polesum}. Then $I[J]$ for the covariance operator $C(\Delta)$ can be written as
\begin{equation}
\label{eqn:integralPoleRep2}
	I[J]
	=
			\sum_{j=1}^{N}	\sum_{n=1}^{k_j}	a_{jn}	\frac{1}{(n-1)!}	\diff{z_j^{n-1}}{^{n-1}} I_{z_j}[J]
	\text{,}
\end{equation}
\end{prop}
\begin{proof}
Substituting the expression for $C(\Delta)$ into $I[J]$ we obtain	\begin{equation}
	I[J]
	=
		\,\langle J, p(\Delta)\, \Theta J\rangle
			+
			\sum_{j=1}^{N}	\sum_{n=1}^{k_j}	a_{jn}	\frac{1}{(n-1)!}	\diff{z_j^{n-1}}{^{n-1}}	\left\langle J,	\Theta	\left(\Delta -z_j\right)^{-1}	J\right\rangle
	\text{.}
\end{equation}
Since $p$ is a polynomial, $p(\Delta)$ is a local operator, thus the inner product $\langle J, \Theta p(\Delta) J \rangle $  vanishes since for $J\in \mathcal{S}(\mathbb{R}^n_{+})$, $\supp J \cap \supp {\Theta J} = \lbrace 0 \rbrace$.
\end{proof}

Having established a connection between reflection positivity and the pole structure of the propagator, we can prove the \emph{if} part of our main theorem.
\begin{prop}
\label{prop:if-part}
A sufficient condition for $C(\Delta)$ to satisfy reflection positivity is that the poles of $C$ all lie on $\R_{-}$, are simple and with non-negative residue.
\end{prop}
\begin{proof}
In \cite{JaRi08,GliJa79,JaRi07I,JaRi07II}, it is shown that, for all $J \in \Sc (\R^d_+)$ the integral $I_{\mu}[J] \geq 0$, $\mu \in \R$, is non-negative.
If all poles are simple and in $\R_-$, we can write \[	I[J]
	=
			\sum_{j=1}^N	a_j	I_{\mu_j}[J]
	\text{,}
\]
with $a_j = \res_{z \to \mu_j}	C(z)$. The claim then follows from the assumption that all residues are non-negative.
\end{proof}

The rest of the paper is devoted to showing that reflection positivity the sufficient condition spelled out in Proposition \ref{prop:if-part} is actually a \emph{necessary} condition. Namely, we will show that reflection positivity is violated whenever there are complex poles, real poles of higher order, or real poles with negative residue.


\section{Necessary conditions for reflection positivity}

\subsection{Properties of the Klein--Gordon quadratic form}
A key role in the discussion of this subsection is played by the Klein--Gordon propagator and its associated integral $I_{z}[J]$ as defined in eq. \eqref{eqn:Izj}.

In this subsection we collect three lemmata that will allow us to construct a function $J\in \Sc(\R^d_{+})$ such that $I_{z}[J]$ and its derivatives strongly violate positivity.

\begin{lemma}
\label{lm:massshell}
	Let $w \in \Cc \setminus \R_+$, $J\in \Sc(\R^d_+)$ and $I_{w}[J]$ be defined as in \eqref{eqn:Izj}.
	The following identity holds true:
	\[\begin{aligned}
	I_{w}[J]
	=&
	-	\pi \imath	\int	\frac{\md^{\sTDim-1}\vec{p}}{\sqrt{{w-\vec{p}\,}^2}}	\,	\tilde{J}^*\left(-\overline{\sqrt{w-{\vec{p}\,}^2}},\vec{p}\right)	\tilde{J}\left(\sqrt{w- {\vec{p}\,}^2},\vec{p}\right)
	\text{.}
	\end{aligned} \]
\end{lemma}
\begin{proof}
	The proof consists of computing a contour integral in $p^0$.
	Factorizing the measure, we obtain after Fourier transforming that
	\[
	\begin{aligned}
	I_{w}[J]
	=
	\int	\md^{\sTDim-1}\vec{p}	\int_{-\infty}^\infty	\md p^0	\,	\frac{\tilde{J}^*(-p^0,\vec{p})	\tilde{J}(p^0,\vec{p})}{(p^0)^2	+	{\vec{p}\,}^2	-	w}
	\text{.}\nonumber
	\end{aligned}
	\]
	The $p^0$ integral can be calculated by closing the contour in the lower half plane. This gives
	\[
	I_{z}[J]
	=
	\lim_{R\to\infty}	\int	\md^{\sTDim-1}\vec{p}	\int_{\Gamma_R}	\md s	\,	\frac{\tilde{J}^*(-\bar{s},\vec{p})	\tilde{J}(s,\vec{p})}{s^2	-(w-	{\vec{p}\,}^2)}
	\text{,}
	\]
	where $\Gamma_R$ is a large semicircle of radius $R$ such that $\Gamma_R\subset\{z\in \mathbb{C} \,|\, \Im(z)\leq 0 \}$. The integrand is meromorphic in the interior of the contour, since $\tilde{J}(s,\vec{p})$ is analytic in the lower half plane by the Paley--Wiener Theorem (cf. \cite[Theorem 7.2.4]{Str03}). Furthermore, since the desired integral over $p^0$ is over the reals, we can extend $\tilde{J}^*(p^0,\vec{p})$ to an analytic function by considering $\tilde{J}^*(\bar{s},\vec{p})$. Finally, the contribution of the semicircular arc of $\Gamma_R$ to the integral vanishes in the limit $R\to \infty$, since $\tilde{J}$ is falling off sufficiently fast by the same Paley--Wiener theorem.

	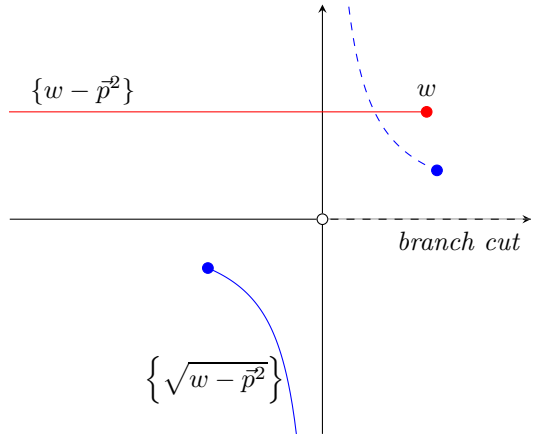
\begin{figure}[b]
		\begin{tikzpicture}
			\begin{axis}[xmin=-3,ymin=-2,xmax=2,ymax=2,axis x line=middle, axis y line=middle, tick align=outside,yticklabels={,,},xticklabels={,,},ytick style={draw=none},xtick style={draw=none}]
				\addplot [domain=0.1:.75*pi,samples=50,color=red]({-cot(\x r)},{1});		
				\filldraw[color=red] (axis cs:1,1) circle (2pt);
				\node at (axis cs:1,1.2) {$w$}; 
				\node at (axis cs:-2.3,1.2) {$\{w-\vec{p}^2\}$}; 
				\filldraw[color=blue] (axis cs:-1.099,-0.455) circle (2pt);
				\addplot [domain=0.1:.75*pi,samples=50,color=blue]({-sin(\x/2 r)/sqrt(sin(\x r))},{-cos(\x/2 r)/sqrt(sin(\x r))});		
				\filldraw[color=blue] (axis cs:1.099,0.455) circle (2pt);
				\addplot [dashed,domain=0.1:.75*pi,samples=50,color=blue]({sin(\x/2 r)/sqrt(sin(\x r))},{cos(\x/2 r)/sqrt(sin(\x r))});		
				\draw[color=white] (axis cs:2,0) -- (axis cs:0,0);
				\draw[dashed] (axis cs:2,0) -- (axis cs:0,0);
				\node at (axis cs:1.3,-.2) {{\it branch cut}};
				\filldraw[color=white] (axis cs:0,0) circle (2pt);
				\draw (axis cs:0,0) circle (2pt);
				\node[anchor=west] at (axis cs:-1.8,-1.5) {$\left\{ \sqrt{w-\vec{p}^2}\right\}$}; 
			\end{axis}
		\end{tikzpicture}
		\caption{The two square roots (in blue) of ${w-\vec{p}\,}^2$ as $\vec{p}$ varies over $\mathbb R^{d-1}$ (in red). We adopt the convention that $\sqrt{ {w-\vec{p}\,}^2}$ lies in the lower half plane (solid blue line). The branch cut will never be crossed because of our assumption that $w \notin \R_+$.}
		\label{fig:branch}
	\end{figure}
	The contour integral is now calculated by the residue theorem; the integrand has poles at $\pm \sqrt{{w-\vec{p}\,}^2}$, with the convention that $\sqrt{{w-\vec{p}\,}^2}$
	lies in the lower half plane (cf. Figure \ref{fig:branch}). The integral is then given by
	\[
	\begin{aligned}
	I_{w}[J]
	=&
	-	2\pi	\imath	\int	\md^{\sTDim-1}\vec{p}	\,	\res_{s\to \sqrt{{w-\vec{p}\,}^2}}	\frac{\tilde{J}^*(-\bar{s},\vec{p})	\tilde{J}(s,\vec{p})}{s^2	-	{w-\vec{p}\,}^2}
	\\=&
	-	\pi \imath	\int	\frac{\md^{\sTDim-1}\vec{p}	}{\sqrt{w-{\vec{p}\,}^2}}	\,	\tilde{J}^*\left(- \overline{\sqrt{{w- \vec{p}\,}^2}},\vec{p}\right)	\tilde{J}\left(\sqrt{w-{\vec{p}\,}^2},\vec{p}\right)
	\text{,}
	\end{aligned} \]
	as desired.
\end{proof}
The second lemma is a homogeneity property of $I_{w}[J]$, which follows from the previous lemma.

\begin{lemma}\label{lm:ihomogen}
	Let $J \in \Sc(\R^d_{+})$, and let $q$ be a polynomial. Then $q(\Delta)J \in \Sc(\R^d_{+})$, and
	\[
	I_{w}[q(\Delta)J]
	=
	\overline{q(\overline{w}) }q(w)		I_{w}[J]
	\text{,}
	\]
	for all $w \in \Cc \setminus \R_+$.
\end{lemma}
\begin{proof}
	Since $q$ is a polynomial, the operator $q(\Delta)$ is local. Thus, the support of $q(\Delta)J$ is contained in the support of $J$, which is contained in $\R^d_{+}$. Furthermore, since $J$ is smooth, $q(\Delta)J$ is smooth as well.

	For the next assertion, we use Lemma \ref{lm:massshell} to calculate $I_w[J]$:
	\[
	\begin{aligned} I_{w}[q(\Delta)J] =   -	\pi \imath	\int	\frac{\md^{\sTDim-1}\vec{p}}{\sqrt{w-{\vec{p}\,}^2}}	\,	\widetilde{q(\Delta)J}^*\left(	- \overline{\sqrt{w-{\vec{p}\,}^2}},\vec{p}\right)	\widetilde{q(\Delta)J}\left(	\sqrt{{w-\vec{p}\,}^2},\vec{p}\right)
	\text{.}
	\end{aligned} \]
	Since $\widetilde{(q(\Delta)J)}(p)= q(p^2) \tilde{J}(p)$ and $\widetilde{(q(\Delta)J)}^*(p)= \overline{q(p^2)} \tilde{J}^*(p)$, the claim follows.
	This completes the proof.
\end{proof}

\begin{lemma}\label{lm:nontriviality}
	Let $I_{w}[J]$ be defined as in eq. \eqref{eqn:Izj}.
	Then there exists a $J\in \Sc(\R^d_{+})$ for which $$I_{w}[J]\neq0\,.$$
\end{lemma}
\begin{proof}
First note that $I_w[J]$ extends to a continuous quadratic form on $L^2(\R_+^d)$, so that by a density argument it suffices to show that there exists a $J \in L^2(\R_+^d)$ for which $I_w[J] \neq 0$. We consider the following explicit candidate 
\[
			J(\tau,\vec{x} ) 
	= 
			\chi_{[a,b]}(\tau)	\cdot	(c_1 \cdots c_{d-1})^{1/4} e^{- \pi( c_1x_1^ 2+ \cdots c_{d-1}x_{d-1}^2)}   
\]
where $\chi_{[a,b]}$ is the indicator function on the interval $[a,b] \subset \R_+$ and $c_1,\ldots,c_{d-1}>0$. One readily checks that for $z$ in the lower half plane we have 
\[
			\widetilde J(z,\vec{p}) 
	=
			\frac{1}{2\pi \imath z}	\left( e^{-2\pi \imath a z} - e^{-2\pi \imath b z} \right)	\frac1{(c_1\cdots c_{d-1})^{1/4}}e^{-\pi (p_1^2/c_1 + \cdots p_{d-1}^2/c_{d-1} )}
	\text{.}
\]
We may now invoke Lemma \ref{lm:massshell} to write for this $J$ that we have 
\begin{align*}
I_w[J]&= -\frac{1}{4 \pi \imath} \int \frac{\md^{\sTDim-1}\vec{p} }{\left({w-\vec{p}\,}^{2}\right)^{3/2}} 
\left( e^{-2\pi \imath a \sqrt{{w-\vec{p}\,}^{2}}}- e^{-2\pi \imath b \sqrt{{w-\vec{p}\,}^{2}}} \right)^2 
\\
&\qquad\qquad\qquad\qquad \times \frac{1}{(c_1\cdots c_{d-1})^{1/2}}  e^{-2 \pi( p_1^2 /c_1+ \cdots p_{d-1}^2 /c_{d-1})}.
\end{align*}
Now observe that as we let $c_1,\ldots c_{d-1} \to 0$ the Gaussian integrals converge to the Dirac delta distribution $\delta(\vec{p})$ (note, however, that this only applies to the combined expression for $\widetilde J^*(-\bar z,p) \widetilde J(z,p)$ appearing in the above integral expression for $I_w(J)$; it is not the case that the function $\widetilde J(z,p)$ itself converges to a Dirac delta distribution). Moreover, it is sufficient to consider this limiting case, since if the above integral is non-zero in this limit, there must exist finite values of $c_1,\ldots,c_{d-1}$ for which $I_w[J]$ is also non-zero. It is now straightforward to compute that in the limit we have
$$
I_w[J]\bigg|_{c_1,\ldots,c_{d-1}\to 0} =  -\frac{1}{2^{(d+3)/2}\pi \imath } w^{-3/2} 
\left( e^{- 2\pi \imath a\sqrt{w}}- e^{-2\pi \imath b \sqrt{w} } \right)^2
$$
which is indeed non-zero for generic values of $a$ and $b$.
\end{proof}
  
\subsection{Reduction to separate poles}
Before proceeding with considering the various cases, we will show that,
whenever we have a term in the sum \eqref{eqn:polesum} which on its own violates reflection positivity, we can always retune $J$ in such a way that it only depends on this term. This follows from Lemma \ref{lm:nontriviality}, combined with the homogeneity property of Lemma \ref{lm:ihomogen}.

\begin{lemma}
\label{lem:multiplePoles}
For any pole $\lambda \in \Pc$, there exists a $J\in \Sc(\R^d_{+})$ such that
\[I[J] = \langle J, C_{\lambda}(\Delta) \Theta J \rangle.\]
\end{lemma}
\begin{proof}

Let again $J \in  \Sc(\R^d_{+})$ be such that $I_{\lambda}[J] \neq 0$, by Lemma \ref{lm:nontriviality}.
We consider the polynomial
\[
			q(z)
	=
			\prod_{z_j \in \Pc \setminus \lbrace \lambda, \overline{\lambda} \rbrace } (z-z_j)^{k_j}
	\text{.}
\]
By definition, \begin{equation}
\label{eq:multisum}
I[q(\Delta)J] \coloneqq \sum_{j}
			\sum_{n=1}^{k_j} a_{jn}	\frac{1}{(n-1)!}	\left.	\diff{z^{n-1}}{^{n-1}}	I_{z}[q(\Delta)J]	\right|_{z=z_j},
\end{equation}
with the first sum running over the set of poles $\Pc$.
If we set $q_j(z) = q(z)/(z-z_j)^{k_j}$, it is easy to see that the sums
\[ \begin{aligned}& \sum_{n=1}^{k_j} a_{jn}	\frac{1}{(n-1)!}		\diff{z^{n-1}}{^{n-1}}	I_{z}[q(\Delta)J]
\\=& \quad
			\sum_{n=1}^{k_j} a_{jn}	\frac{1}{(n-1)!}			\diff{z^{n-1}}{^{n-1}}(z-z_j)^{k_j}	(z-\overline{z_j})^{k_j}	I_s[q_j(\Delta)J]
	\end{aligned}
\]
vanish upon evaluation at  $z_j \neq \lambda, \overline{\lambda}$.

Therefore, the only terms that are left in the sum \eqref{eq:multisum} are the ones corresponding to $\lambda, \overline{\lambda}$, giving
\[\begin{aligned}
			I[q(\Delta)J]
	=& \left\langle q(\Delta) J, C_{\lambda}(\Delta) \Theta q(\Delta) J \right\rangle \text{,}
	\end{aligned}
\]
which proves the claim.
\end{proof}

\subsection{Violation of reflection positivity}
We are now ready to prove that the condition in Propostion \ref{prop:if-part} is also necessary. We do this by constructing functions in $\Sc(\R_+^d)$ that violate reflection positivity, in all the cases left out by the sufficient condition in Proposition \ref{prop:if-part}. 
The case of complex poles is covered by Propostion \ref{prop:complexRPVio} while poles of higher order are excluded in Proposition \ref{prop:RealHigOrd}.

\begin{prop}\label{prop:complexRPVio}
Let us assume that $C$ has a pole in the upper half-plane.
Then, there exists $J\in \Sc(\R^d_{+})$ such that
\begin{align}
I[J] <0 \text{.}
\end{align}
\end{prop}
\begin{proof}
In view of Lemma \ref{lem:multiplePoles}, we may assume, without loss of generality, that $C$ has exactly two complex conjugate poles, $\lambda$ and $\overline{\lambda}$ of order $k$, with $\Im( \lambda) >0$.

Let $J \in \Sc(\R^d_+)$ be such that $I_{\lambda}[J] \neq 0$, which is always possible in view of Lemma \ref{lm:nontriviality}.
Let us choose a polynomial ansatz for $q$ defined as follows
\[
 q(z)
 =
 (z-\lambda)^{k-1}	h(z)
 \text{,}
 \quad
 h(s)
 =
 \frac{1
  -	\bar{\alpha}}{2\imath	\Im(\lambda)}	(z-\lambda)
 +	1
 \text{.}
\]
Note that the polynomial $h$ is chosen in such a way that $h(\lambda) = 1$, and $\overline{h(\overline{\lambda})} = \alpha$.

We now compute
\[
\begin{aligned} I[q(\Delta)J]
 =&
 2\Re	\sum_{n=1}^{k}	a_n	\frac{1}{(n-1)!}	\left.
 \diff{z^{n-1}}{^{n-1}}	I_{z}[q(\Delta)J]
 \right|_{z=\lambda}
 \\=&
 2\Re	\sum_{n=1}^{k}	a_n	\frac{1}{(n-1)!}	\left.
 \diff{s^{n-1}}{^{n-1}}	(z-\lambda)^{k-1}	(z-\bar{\lambda})^{k-1}	h(z)	\overline{h(\bar{z})}	I_{z}[J]
 \right|_{z=\lambda}
 \\=&
 2\Re	\left(
 \alpha \cdot a_{k}	\left(2\imath \Im(\lambda)\right)^{k-1}	I_{\lambda}[J]
 \right)
 \text{,} \nonumber
\end{aligned}\]
In the first step, we have applied Lemma \ref{lm:ihomogen}. In the last step, we have used that $(z-\lambda)^{k-1}$ vanishes at $\lambda$, with all of its derivatives except for the $k-1$-th. Note that $\Im(\lambda) \neq 0$, and $I_{\lambda}[J] \neq 0$ by construction, hence we can choose $\alpha$ to make the above 
expression negative, to prove the statement.
\end{proof}

The case of real poles of order greater than one can be treated similarly.
\begin{prop}\label{prop:RealHigOrd}
Suppose that $C(z)$ has one real pole of order greater than one.
Then, there exists $J\in \Sc(\R^d_+)$ such that
\begin{align}
I[J] < 0.
\end{align}
\end{prop}
\begin{proof}
By Lemma \ref{lem:multiplePoles}, it is not restrictive to assume that $C$ has exactly one pole $\mu$ of order $k>1$.

Let $J \in \Sc(\R^d_{+})$ be such that $I_{\mu}[J] \neq 0$, which again is always possible in view of Lemma \ref{lm:nontriviality}.

We proceed similarly to the proof of Proposition \ref{prop:complexRPVio} and choose for $q$ a polynomial of the form
\[
 q(z)
 =
 \alpha	(z-\mu)^{k-1}
 +	1
 \text{.}
\]
Computing $I[q(\Delta)J]$, we obtain
\[
\begin{aligned}& I[q(\Delta)J]
 =
 \sum_{n=1}^{k}	a_n	\frac{1}{(n-1)!}	\left.
 \diff{z^{n-1}}{^{n-1}}	I_{z}[q(\Delta)J]
 \right|_{z=\mu}
 \\& \quad =
  \sum_{n=1}^{k}	a_n	\frac{1}{(n-1)!}	\left.
 \diff{z^{n-1}}{^{n-1}}\left(
 1
 +	2	\alpha	(z-\mu)^{k-1}
 +	\alpha^2	(z-\mu)^{2(k-1)}
 \right)	I_{z}[J]
 \right|_{z=\mu}
 \\& \quad =
 I[J]
 +	2	\alpha	\cdot a_k	I_{\mu}[J]
 \text{,}
\end{aligned}
\]
where we have used Lemma \ref{lm:ihomogen} in the first step.
In the second step, we observed that all derivatives of $(z-\mu)^{k}$ evaluated at $\mu$ vanish, except for the $(k-1)$-th.
Note that $a_{k}$ is non zero by assumption and $I_{\mu}[J]$ is non zero by construction; thus, we can choose $\alpha$ such that $I[q(\Delta)J]<0$,
which proves the claim.
\end{proof}
Finally, whenever $C$ has a real single pole with negative residue, we can use Lemma \ref{lem:multiplePoles} to find a $J$ such that $I[J]$ only depends on that pole, and the sign of the residue will automatically imply that reflection positivity is violated.
\medskip

Combined with the proof that simple real poles with non-negative residue satisfy reflection positivity, we have completed the proof of our main result:
\begin{theorem}
\label{thm:main}
	Let $C$ be a real rational function which has no poles on $\Real_+$.
A necessary and sufficient condition for the operator $C(\Delta)$ to satisfy reflection positivity is that the poles of $C$ all lie on $\R_{-}$, are simple and with non-negative residue.
\end{theorem}

\end{document}